\documentclass{easychair}

\usepackage[utf8]{inputenc}
\usepackage[T1]{fontenc}

\usepackage{mathrsfs} 
\usepackage{amsmath}
\usepackage{amsfonts}
\usepackage{amssymb}

\usepackage{hyperref}

\usepackage{todonotes}

\usepackage{graphicx}
\usepackage{tikz}
\usetikzlibrary{calc, shapes, positioning}

\tikzset{
	vertex/.style=
	{
		circle,
		minimum size = 3mm,
		draw=black
	}
}
\tikzset{
	knoten/.style=
	{
		circle,
		inner sep=1pt,
		draw=black
	}
}
\tikzset{
	description/.style=
	{
		draw = none
	}
}
\tikzset{
	dot/.style=
	{
		circle,
		minimum width = 1mm,
		draw=black,
		fill=white
	}
}
\tikzset{
	smalldot/.style=
	{
	circle,
		inner sep = 1pt,
		draw=black,
		fill=white
	}
}
\tikzset{
	root/.style=
	{
		circle,
		minimum width = 1mm,
		draw=black,
		fill=black!50
	}
}

\tikzset{->, >=stealth, auto, node distance=2.8cm, semithick}

\tikzset{triplearrow/.style={draw=black, color=black, thick, double distance=4pt, -, >=stealth}, thirdline/.style={draw=black!75, color=black!75, thick, -, >=stealth}}

\usetikzlibrary{positioning,shapes,shadows,arrows}

\tikzstyle{abstract}=[rectangle, draw=black, rounded corners, fill=blue!40, drop shadow,
        text centered, anchor=north, text=white, text width=3cm]
\tikzstyle{comment}=[rectangle, draw=black, rounded corners, fill=green, drop shadow,
        text centered, anchor=north, text=white, text width=3cm]
\tikzstyle{myarrow}=[->, >=open triangle 90, thick]
\tikzstyle{line}=[-, thick]

\tikzstyle{poldelay}=[rectangle, draw=black, rounded corners, fill=green!60, text centered, anchor=north, text=black, text width=3cm]

\tikzstyle{newdelay}=[rectangle, draw=black, line width=1mm, rounded corners, fill=green!60, text centered, anchor=north, text=black, text width=3cm]

\tikzstyle{newdelaytree}=[rectangle, draw=black, line width=1mm, rounded corners, fill=blue!40, text centered, anchor=north, text=black, text width=3cm]

\tikzstyle{incdelay}=[rectangle, draw=black, rounded corners, fill=orange!80, text centered, anchor=north, text=black, text width=3cm]

\tikzstyle{nopt}=[rectangle, draw=black, rounded corners, fill=red!90, text centered, anchor=north, text=black, text width=3cm]

\tikzstyle{legend}=[rectangle, draw=black, rounded corners, text centered, anchor=north, text=black, text width=3.5cm]
\tikzstyle{newdelaytreelegend}=[rectangle, draw=black, rounded corners, fill=blue!40, text centered, anchor=north, text=black, text width=3cm]
\tikzstyle{newdelaylegend}=[rectangle, draw=black, line width=1mm, rounded corners, fill=white, text centered, anchor=north, text=black, text width=3cm]

\tikzstyle{subclass}=[->, >=open triangle 45] 
 
\usepackage{amsthm}
\newtheorem{definition}{Definition}[section]
\newtheorem{lemma}[definition]{Lemma}
\newtheorem{theorem}[definition]{Theorem}

\newtheorem{corollary}[definition]{Corollary}

\newcommand{\dhyp}[1]{\textit{#1}}

\newcommand{\Set}[1]{\{#1\}}
\newcommand{\suchthat}[0]{ \ : \ }

\newcommand{\abs}[1]{\left| #1 \right|}

\newcommand{\NPcomplete}[0]{\textbf{NP}-complete}

\begin{document}

\title{Simple Necessary Conditions for the Existence of a Hamiltonian Path with Applications to Cactus Graphs}

\author{Pascal Welke}
\institute{Informatik III, University of Bonn, Germany}

\authorrunning{Pascal Welke}
\titlerunning{Necessary Conditions for Hamiltonian Paths}

\maketitle

\begin{abstract}
We describe some necessary conditions for the existence of a Hamiltonian path in any graph (in other words, for a graph to be traceable).
These conditions result in a linear time algorithm to decide the Hamiltonian path problem for cactus graphs.
We apply this algorithm to several molecular databases to report the numbers of graphs that are traceable cactus graphs.
\end{abstract}

\section{Introduction}
A Hamiltonian path is a path in a graph $G$ that contains each vertex of $G$ exactly once.
The Hamiltonian path problem (i.e., does there exist a Hamiltonian path in a given graph $G$?) is a well studied \NPcomplete{} problem with various applications~\cite{garey1979computers}.
Several algorithms have been proposed to find a Hamiltonian path in a graph, or to decide that none exists.
For example, Held and Karp~\cite{held1962dynamic} give a $O(n^2 \cdot 2^n)$ algorithm to compute a Hamiltonian path.
Bj{\"o}rklund~\cite{bjorklund2014determinant} gives a $O(1.657^n)$ time algorithm to count the number of Hamiltonian paths in a graph, which can also be used to decide the Hamiltonian path problem.
Due to the exponential time complexity of those and other algorithms, it would be beneficial to derive simple, fast tests that can be run in advance to decide at least in some cases if there exists a Hamiltonian path, or not.

Many authors concentrated on sufficient conditions for a graph to be traceable (i.e., that it contains a Hamiltonian path).
E.g. Dirac~\cite{dirac1973note} gives a lower bound on the number of edges in a graph that implies the existence of a Hamiltonian path.
Also, there is a wide range of graph classes, where we know that a Hamiltonian path exists, e.g. complete graphs, cycles, paths, or graphs of the platonic solids.

We go a different way and consider situations which do not allow for a Hamiltonian path.
That is, we define easily verifiable properties of graphs that prove that a graph is not traceable.
To our knowledge, there is much less work in this direction, most notably by Chv{\'a}tal~\cite{chvatal1973edmonds}, that introduces weakly Hamiltonian graphs and derives necessary conditions for a graph to contain a Hamiltonian cycle.
However, the paper uses quite involved concepts and the verification of the conditions for a given graph is not straightforward.
Our conditions, on the other hand, can be checked in linear time and are easy to understand.
They are based on partitioning a graph $G$ into its biconnected components and deriving a tree structure from those objects.
In short, a Hamiltonian path in $G$ can only exist if this tree structure is a path.

We start by considering trees and continue by defining a tree structure using the biconnected components of an arbitrary graph to devise conditions in Lemmas~\ref{lem:criticality} and~\ref{lem:bics}.
As a direct application of our necessary conditions, we devise a linear time algorithm for cactus graphs in Theorem~\ref{thm:tracablecactus}. 
Finally, we give statistics of a molecular dataset that were obtained using our conditions.

\section{Nice Necessary Conditions}\label{sec:conditions}
From now on, we only consider connected graphs, as otherwise there cannot be a Hamiltonian path.
We start by considering the Hamiltonian Path problem for trees. 
It is easy to see, that a tree $T$ has a Hamiltonian Path if and only if $T$ is a path.
We use standard graph notation, see Diestels book~\cite{diestel2005graph} for definitions.

\begin{lemma}\label{lem:trees}
A tree $T$ has a Hamiltonian path if and only if $T$ is a path.
\end{lemma}

\begin{proof}
``$\Leftarrow$'' is clear.
``$\Rightarrow$'' Let $T$ be a tree and $P$ a Hamiltonian path in $T$.
$P$ contains all vertices of $T$ and has thus $\abs{V(G)} - 1$ edges. 
Therefore, $E(T) = E(P)$ and thus $T$ is a path.
\end{proof}

We will show that a generalized version of this holds for a tree structure defined on the \emph{articulation} vertices of any graph $G$.
We need the following definition:

\begin{definition}
Let $G$ be a connected graph.
A vertex $v \in V(G)$ is called \dhyp{articulation} vertex if its removal disconnects $G$, 
i.e., the graph $G-v = (V', E')$ is disconnected, where $V' := V \setminus \Set{v}$ and $E' := \Set{e \in E \suchthat v \notin e}$.
The \dhyp{criticality} of $v$ is the number of connected components of $G-v$.
\end{definition}

In a tree, every vertex that is not a leaf is an articulation vertex.
We now prove the first necessary condition.
In the case of trees, it follows directly from Lemma~\ref{lem:trees}.

\begin{lemma}\label{lem:criticality}
Let $G$ be a traceable graph.
Then all vertices have criticality at most $2$.
\end{lemma}

\begin{proof}
Suppose there is a vertex $v$ with criticality at least $3$.
Then $G-v$ has three nonempty connected components $C_1, C_2, C_3$.
Let $P$ be a Hamiltonian path of $G$ and $u_1$ (resp. $u_2, u_3$) be the first vertex in $V(C_1)$ (resp. $V(C_2), V(C_3)$) occurring in $P$ (w.l.o.g. in this order).
Any path connecting $u_1 \in V(C_1)$ to $u_2 \in V(C_2)$ in $G$ needs to contain $v$.
Otherwise, $u$ and $w$ would be contained in the same connected component of $G-v$.
The same is true for a path from $u_2$ to $u_3$.
Therefore, $P$ contains $v$ at least twice, which is a contradiction to $P$ being a path.
\end{proof}

Figure~\ref{fig:lem:bics} shows an illustration of the situation described in Lemma~\ref{lem:criticality}.
Vertex $v_2$ has criticality $3$ and therefore does not allow for a Hamiltonian path in the graph.
The next lemma focuses on biconnected components.

\begin{lemma}\label{lem:bics}
Let $G$ be a traceable graph.
Then each biconnected component of $G$ contains at most two articulation vertices.
\end{lemma}

\begin{proof}
Suppose there is a biconnected component $B$ of $G$ that contains three articulation vertices $v_1, v_2, v_3$.
Removing $v_i \in \Set{v_1, v_2, v_3}$ from $G$ results in a disconnected graph $G_i := G - v_i$.
Now, there exists a connected component $B_i$ in $G_i$ such that $V(B_i) \cap V(B) = V(B) \setminus \Set{v_i}$ and $B_i$ is connected.
Let $X_i$ be the nonempty graph of all other connected components of $G-v_i$.
To see this, remember that $B$ is a biconnected component thus removing a single vertex does not disconnect $B$ and all vertices in $V(B) \setminus \Set{v_i}$ are contained in the same connected component of $G-v_i$.
However, as $v_i$ is an articulation vertex, $G - v_i$ is disconnected and thus $V(X_i) \neq \emptyset$.
As an example, Figure~\ref{fig:lem:bics} shows $B_i$ and $X_i$ for the case $v_i = v_1$.

\emph{Claim:} 
$V(X_i) \cap V(X_j) = \emptyset$ for all $i \neq j \in \Set{1,2,3}$.

Using this claim, we can prove the lemma.
A Hamiltonian path $P$ of $G$ needs to contain all vertices in $V(X_1), V(X_2), V(X_3)$.
But to get from any vertex in $V(X_i)$ to a vertex $x \in V(X_j)$, it needs to pass through $v_i$.
To get from $v_i$ to $x$, the path must pass through $v_j$, as $v_i \in V(B_j)$.
Using the same argument as in the proof of Lemma~\ref{lem:criticality}, we see that $P$ needs to visit one of the articulation vertices $v_1, v_2, v_3$ at least twice, which is a contradiction to $P$ being a path.

\emph{Proof of Claim:} 
Suppose there exists $x \in V(X_i) \cap V(X_j)$.
As $x \in V(X_i)$ there exists a path in $X_i$ connecting $x$ to a neighbor of $v_i$ in $G$.
Thus removing $x_j$ would not disconnect $x$ from $v_i \in V(B_j)$, which contradicts $x \in V(X_j)$. 
\end{proof}

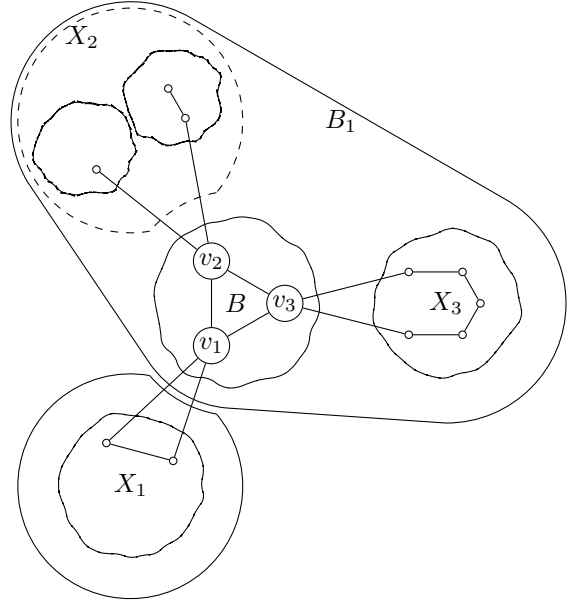
\begin{figure}
\begin{center}

\newcommand\irregularcircle[2]{
  \pgfextra {\pgfmathsetmacro\len{(#1)+rand*(#2)}}
  +(0:\len pt)
  \foreach \a in {10,20,...,350}{
    \pgfextra {\pgfmathsetmacro\len{(#1)+rand*(#2)}}
    -- +(\a:\len pt)
  } -- cycle
}

\begin{tikzpicture}[scale=0.65]

\pgfmathsetseed{42}

\pgfmathsetmacro{\cyclesize}{2}
\pgfmathsetmacro{\innercutsize}{2.3}
\pgfmathsetmacro{\vdist}{1}
\pgfmathsetmacro{\turn}{0}
\pgfmathsetmacro{\half}{(abs(\cyclesize - \innercutsize)) / 2}
\pgfmathsetmacro{\mediumsize}{\cyclesize + \half}
\pgfmathsetmacro{\centerdist}{\cyclesize + \innercutsize}
\pgfmathsetmacro{\innercutmed}{\innercutsize + \half}

\coordinate (center) at (0,0);

\draw[rounded corners=1mm] (center) \irregularcircle{1.7cm}{1mm};
\node [knoten] (v1) at ($ (center) + ( 0    : \vdist cm) $) {$v_3$};
\node [knoten] (v2) at ($ (center) + ( 120  : \vdist cm) $) {$v_2$};
\node [knoten] (v3) at ($ (center) + ( 240  : \vdist cm) $) {$v_1$};
\node [description] (d1) at (center) {$B$};
\draw [-] (v1) edge (v2)
          (v2) edge (v3)
          (v3) edge (v1);

\coordinate (c1) at ($ (center) + (0 : \centerdist cm) $);
\draw[rounded corners=1mm] (c1) \irregularcircle{1.5cm}{1mm};
\node [smalldot] (v11) at ($ (c1) + (140 : 1cm) $) {};
\node [smalldot] (v13) at ($ (c1) + (220 : 1cm) $) {};
\draw [-] (v1) edge (v11)
		  (v1) edge (v13);
\node [description] (d1) at (c1) {$X_3$};

\node [smalldot] (v14) at ($ (v11) + (0 : 1.1cm) $) {};
\node [smalldot] (v15) at ($ (v13) + (0 : 1.1cm) $) {};
\node [smalldot] (v16) at ($ (v1)  + (0 : 4  cm) $) {};
\draw [-] 
          (v11) edge (v14)
          (v13) edge (v15)
          (v14) edge (v16)
          (v15) edge (v16);

\draw [-, dashed] ($ (center) + ( + 100 : \innercutsize cm) $) arc ( + 100 : + 140 : \innercutsize cm );
\draw [-, dashed] ($ (center) + ( + 140 : \innercutsize cm) $) arc ( + 280 : - 40  : \innercutsize cm );

\coordinate (c2) at ($ (center) + (120 : \centerdist cm) $);
\coordinate (c21) at ($ (c2) + (30 : 1cm) $);
\coordinate (c22) at ($ (c2) + (210 : 1.1cm) $);
\draw[rounded corners=1mm] (c21) \irregularcircle{1cm}{1mm};
\draw[rounded corners=1mm] (c22) \irregularcircle{1cm}{1mm};
\node [smalldot] (v21) at ($ (c21) + (-60 : 0.5 cm) $) {};
\node [smalldot] (v22) at ($ (c22) + (-60 : 0.5 cm) $) {};
\draw [-] (v2) edge (v21)
		  (v2) edge (v22);
\node [description] (d2) at ($ (c2) + (120 : \cyclesize cm) $) {$X_2$};

\node [smalldot] (v23) at ($ (v21) + (120 : 0.7cm) $) {};
\draw [-] (v21) edge (v23);

\draw [-] ($ (center) + ( + 220 : \innercutsize cm) $) arc ( + 220 : + 260 : \innercutsize cm );
\draw [-] ($ (center) + ( + 260 : \innercutsize cm) $) arc ( + 40 : - 280  : \innercutsize cm );

\coordinate (c3) at ($ (center) + (240 : \centerdist cm) $);
\draw[rounded corners=1mm] (c3) \irregularcircle{1.5cm}{1mm};
\node [smalldot] (v31) at ($ (c3) + (30 : 1cm) $) {};
\node [smalldot] (v32) at ($ (c3) + (120 : 1cm) $) {};
\draw [-] (v3) edge (v31)
		  (v3) edge (v32);
\node [description] (d3) at (c3) {$X_1$};

\draw [-] (v31) edge (v32);

\draw [-] ($ (center) + ( + 215 : \mediumsize cm) $) arc ( + 215 : + 265 : \mediumsize cm ) -- ($ (c1) + (270 : \innercutmed cm) $) arc (-90 : 60 : \innercutmed cm) -- ($ (c2) + (60 : \innercutmed cm) $) arc (60 : 210 : \innercutmed cm) -- cycle; 

\node [description] (d4) at ($ (center) + (60 : \centerdist) $) {$B_1$};

\end{tikzpicture} \end{center}
\caption{A cactus graph $G$ without a Hamiltonian path.
$v_2$ has criticality $3$ (Lemma~\ref{lem:criticality}) and the biconnected component $B$ contains three articulation vertices (Lemma~\ref{lem:bics}).}
\label{fig:lem:bics}
\end{figure}

Lemma~\ref{lem:criticality} and Lemma~\ref{lem:bics} together show that on any graph $G$, the existence of a Hamiltonian path implies a path-structure on the biconnected components of $G$.
More exactly, let $\mathcal{A}(G)$ be the set of articulation vertices of $G$ and $\mathcal{B}$ be the set of biconnected components of $G$.
We define a new graph $A(G) = (\mathcal{A}(G), E')$ where $E' = \Set{\Set{v,w} \suchthat \exists B \in \mathcal{B}:v, w \in V(B)}$.
Then $A(G)$ is a path.
Therefore, the Hamiltonian path problem reduces to checking if these two conditions hold and if there is a Hamiltonian path in each biconnected component, that
\begin{itemize}
\item starts at the first articulation vertex and ends at the second articulation vertex (if there are two)
\item starts at the articulation vertex (if there is one)
\item starts and ends at arbitrary vertices (if there is no articulation vertex in $G$).
\end{itemize}

We call biconnected components that contain exactly one articulation vertex \dhyp{leaf components} and finish this section with an easy corollary of the above considerations.

\begin{corollary}\label{corr:leaves}
Let $G$ be a traceable graph.
Then there are either zero or two leaf components.
\end{corollary}

\section{The Hamiltonian Path Algorithm for Cactus Graphs}
The results of Section~\ref{sec:conditions} imply a polynomial time algorithm for the Hamiltonian path problem for cactus graphs.
A \dhyp{cactus graph} is a connected graph where every biconnected component is either a single edge or a simple cycle.
Figure~\ref{fig:cactus} shows a cactus graph and a graph that is no cactus.

\begin{figure}
\begin{center}

\begin{tikzpicture}[scale=0.8]

\pgfmathsetmacro{\vdist}{1}

\coordinate (center) at (0,0);

\node [knoten] (a) at ($ (center) + (60 : \vdist cm) $) {};
\node [knoten] (b) at ($ (center) + (120 : \vdist cm) $) {};
\node [knoten] (c) at ($ (center) + (180 : \vdist cm) $) {};
\node [knoten] (d) at ($ (center) + (240 : \vdist cm) $) {};
\node [knoten] (e) at ($ (center) + (300 : \vdist cm) $) {};
\node [knoten] (f) at ($ (center) + (360 : \vdist cm) $) {};

\draw [-] (a) edge (b)
          (b) edge (c)
          (c) edge (d)
          (d) edge (e)
          (e) edge (f)
          (f) edge (a);

\node [knoten] (g) at ($ (f) + (30 : \vdist cm) $) {};
\node [knoten] (h) at ($ (f) + (-30 : \vdist cm) $) {};

\draw [-] (f) edge (g)
          (g) edge (h)
          (h) edge (f);

\node [knoten] (i) at ($ (c) + ( 180 : \vdist cm) $) {};
\node [knoten] (j) at ($ (g) + ( 0 : \vdist cm) $) {};
\draw [-] (c) edge (i)
          (g) edge (j);

\coordinate (center) at (7,0);

\node [knoten] (a) at ($ (center) + (60 : \vdist cm) $) {};
\node [knoten] (b) at ($ (center) + (120 : \vdist cm) $) {};
\node [knoten] (c) at ($ (center) + (180 : \vdist cm) $) {};
\node [knoten] (d) at ($ (center) + (240 : \vdist cm) $) {};
\node [knoten] (e) at ($ (center) + (300 : \vdist cm) $) {};
\node [knoten] (f) at ($ (center) + (360 : \vdist cm) $) {};

\draw [-] (a) edge (b)
          (b) edge (c)
          (c) edge (d)
          (d) edge (e)
          (e) edge (f)
          (f) edge (a)
          (e) edge (b);

\node [knoten] (g) at ($ (d) + (180:\vdist cm) $) {};
\draw [-] (g) edge (d);

\coordinate (newcenter) at ($ (f) + (0:\vdist cm) $);

\node [knoten] (a) at ($ (newcenter) + (60 : \vdist cm) $) {};
\node [knoten] (b) at ($ (newcenter) + (120 : \vdist cm) $) {};
\node [knoten] (c) at ($ (newcenter) + (180 : \vdist cm) $) {};
\node [knoten] (d) at ($ (newcenter) + (240 : \vdist cm) $) {};
\node [knoten] (e) at ($ (newcenter) + (300 : \vdist cm) $) {};
\node [knoten] (f) at ($ (newcenter) + (360 : \vdist cm) $) {};

\draw [-] (a) edge (b)
          (b) edge (c)
          (c) edge (d)
          (d) edge (e)
          (e) edge (f)
          (f) edge (a)
          (d) edge (b)
          (e) edge (a);

\node [knoten] (g) at ($ (a) + (0:\vdist cm) $) {};
\draw [-] (g) edge (a);

\end{tikzpicture} \end{center}
\caption{A cactus graph on the left and a graph that is not a cactus on the right.}
\label{fig:cactus}
\end{figure}
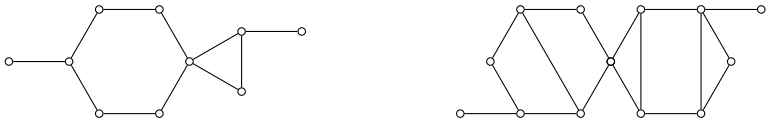

\begin{theorem}\label{thm:tracablecactus}
A cactus graph is traceable if and only if all of the following three conditions hold:
\begin{itemize}
\item Each vertex has criticality at most two
\item Each biconnected component contains at most two articulation vertices
\item If a biconnected component contains two articulation vertices, they share an edge.
\end{itemize}
\end{theorem}

\begin{proof}
Each cycle is traceable, and each Hamiltonian path of a cycle $C$ starts at an arbitrary vertex of $C$ and ends at one of its two neighbors.
Edges are also traceable.
``$\Rightarrow$''
If a cactus graph $G$ is traceable then, by Lemmas~\ref{lem:criticality} and~\ref{lem:bics} the first two conditions hold.
Let $B$ be a biconnected component of $G$ that contains two articulation vertices.
If $B$ is an edge, then the third condition holds trivially.
If $B$ is a cycle, then any Hamiltonian path must enter $B$ through one articulation vertex $v$, leave it through the other $w$ and can never enter $B$ again.
Therefore, the path from $v$ to $w$ must be a Hamiltonian path of $B$ and therefore contains all edges in $E(B)$ except one, which must be $\Set{v,w}$. 
``$\Leftarrow$''
We construct a Hamiltonian path as follows:
If $G$ is biconnected (i.e., it has no articulation vertices), we construct a Hamiltonian path by removing an arbitrary edge.
Otherwise, for each cycle, we remove the edge between the two articulation vertices or one of the edges incident to the unique articulation vertex in the cycle.
Note that by this, each articulation vertex has degree two in the resulting graph $P$.
As vertices with criticality $0$ have degree one or two in $G$, every vertex in $P$ has degree less than three.
We have removed exactly one edge from each cycle of $G$, thus $P$ contains no cycles and is still connected.
Therefore, $P$ is a path.
\end{proof}

We can check the conditions of Theorem~\ref{thm:tracablecactus} in linear time for a graph $G$ as follows:
First, we check if $G$ is connected by a simple breadth first search in linear time.
Next, we compute the biconnected components of $G$ in linear time using Tarjans algorithm~\cite{tarjan1972depth}.
Having the biconnected components (given as lists of edges), it is easy to compute the criticality of each vertex in $G$ by counting the number of biconnected components each vertex occurs in as an endpoint of at least one edge.
Having the criticality of each vertex, we can compute the number of critical vertices per biconnected components by a single pass over its edge list.
To check if $G$ is a cactus graph, we test if each biconnected component is either an edge or a simple cycle, which can also be done by a single pass over all edges in a biconnected component.
If there are exactly two, by another pass we can check if the component contains an edge that contains both critical vertices.
Therefore, the algorithm can be implemented to run in linear time with a small constant.

\section{Some Statistics for Molecular Datasets}
We implemented some variants of the proposed algorithm and applied them to three well studied molecular datasets.
\begin{description}
\item[NCI-HIV] consists of almost $43k$ compounds that are annotated with their activity against the human immunodeficiency virus (HIV) provided by the National Cancer institute~\cite{ncihiv}. 
We do not consider the annotations here, but merely use the molecular graph representations.
The median number of vertices and edges per graph are $41$ and $43$, respectively.
The database consumes $20.1MB$ in our textual file format.
\item[NCI-2012] is a larger set of molecular graphs from the same source~\cite{ncihiv}.
It consists of more than $250k$ graphs with median number of vertices and edges $36$ and $37$, respectively.
The file size is $100.3MB$.
\item[ZINC] is a subset of almost $9$ million so called 'Lead-Like' molecules from the zinc database of purchasable chemical compounds~\cite{zincll}.
The molecules were selected to have molar mass between $250g/mol$ and $350g/mol$ and have median number of vertices and edges $43$ and $44$, respectively.
The file size is roughly $3.8GB$.
\end{description}

Figure~\ref{tab:statistics} shows the number of graphs $N$, the number of connected cactus graphs $C$, the number of traceable cactus graphs $T$, as well as the number of (arbitrary) graphs $X$ that are definitively not traceable.
Furthermore, it reports the time $t_i$ needed by our implementation to compute value $i \in \Set{N,C,T,X}$.
The numbers were computed by parsing the database from a text file and checking property $i$ for each graph in the respective database.
Times were measured using the GNU \verb+time+ command summing up \verb+sys+ and \verb+user+ times.

All experiments were done on an Intel Core i7-2600K with 8GB main memory running Ubuntu 14.04 64bit.
The algorithms were implemented in C and compiled using \verb+gcc 4.8.2+ with optimization flag \verb+-O3+ enabled.
No multi-threading was used.
Furthermore, due to the fact that each graph can be processed separately, the maximum memory consumption at any time was less than $10MB$.

$t_N$ reports the time our implementation needs to parse the graph database, create graph objects in memory, and dump them again.
As you can see, the actual tests only add a small overhead in time compared to just parsing the data.
On the other hand, by checking if a graph (a) is connected and (b) fulfills our two necessary conditions, we can declare most of the graphs from all databases as non traceable.
For example, for the ZINC dataset, we would only need to further investigate $7$ out of almost $9$ million graphs to check if they are traceable, or not.

\begin{table}
\begin{center}
\begin{tabular}{ r | c c c}
      & NCI-HIV & NCI-2012 & ZINC \\
  \hline
  $N$ & 42687 & 249533 & 8946757 \\
  $C$ & 18028 & 134478 & 6517109 \\
  $X$ & 42658 & 249436 & 8946750 \\
  $T$ &     6 &     80 &       0 \\
  $t_N$ & 0.20 & 1.01 & 39.37 \\
  $t_C$ & 0.28 & 1.47 & 56.50 \\
  $t_X$ & 0.31 & 1.57 & 63.94 \\
  $t_T$ & 0.32 & 1.67 & 70.31 \\
\end{tabular}
\end{center}
\caption{Statistics for three molecular data sets. 
The reported times are in seconds}
\label{tab:statistics}
\end{table}

\section{Conclusion}
We have proposed two necessary conditions for a graph to be traceable that are easy and fast to check.
Using them, we proposed a linear time algorithm that decides if a cactus graph is traceable.
In more general practical settings, checking these conditions could be a first step, that might, in many cases make applying one of the exponential time exact algorithms obsolete.
We evaluated our tests effectiveness in that respect on three molecular data sets of varying size and showed that most molecular graphs can be easily identified as non-traceable using our conditions.

In future work, our proposed algorithm can be extended to yield an exact polynomial time algorithm for more general classes of graphs.
Using our conditions, we can reduce the Hamiltonian path problem in a non-biconnected graph $G$ to smaller Hamiltonian path problems in the biconnected components of $G$.
We would only need to check if there is a Hamiltonian path in each biconnected component that connects the two articulation vertices or starts at the unique articulation vertex, respectively.
This is possible in polynomial time if, for example, the number of spanning trees in each biconnected component is bounded by a polynomial in the size of $G$.

\bibliographystyle{plain}

\end{document}